\documentclass[journal]{IEEEtran}
\ifCLASSINFOpdf
\else
\fi
\usepackage{tikz}
\usepackage{amsmath}
\usepackage{amsfonts}
\usepackage{stfloats}
\usepackage{amsthm}
\usepackage{amssymb}
\usepackage{mathtools}  
\usepackage{makecell}
\usepackage{booktabs}
\usepackage{multirow}

\newtheorem{theorem}{Theorem}[section]   

\usepackage{tikz,xcolor,hyperref}
\definecolor{darkgreen}{RGB}{0,100,0} 
\definecolor{lime}{HTML}{A6CE39}
\DeclareRobustCommand{\orcidicon}{%
 \begin{tikzpicture}
 \draw[lime, fill=lime] (0,0) 
 circle [radius=0.16] 
 node[white] {{\fontfamily{qag}\selectfont \tiny ID}}; \draw[white, fill=white] (-0.0625,0.095) 
 circle [radius=0.007]; \end{tikzpicture}
 \hspace{-2mm}}
\foreach \x in {A, ..., Z}{%
 \expandafter\xdef\csname orcid\x\endcsname{\noexpand\href{https://orcid.org/\csname orcidauthor\x\endcsname}{\noexpand\orcidicon}}
 }

\begin{document}
%
\title{PEEL: A Poisoning-Exposing Encoding Theoretical Framework for Local Differential Privacy}
%
%
%

\author{Lisha~Shuai\orcidA{}, 
                 Jiuling~Dong\orcidD{}, 
                 Nan~Zhang\orcidG{}, 
                 Shaofeng~Tan\orcidB{}, 
                 Haokun~Zhang\orcidC{}, 
                 Zilong Song\orcidF{},
                 Gaoya Dong\orcidH{}, 
                 and Xiaolong~Yang\orcidE{}, \textit{Member, IEEE}
\thanks{This work was supported by the National Natural Science Foundation of China under Grant 61971033, Beijing Natural Science Foundation General Project under Grant 4232016, and Guangdong Basic and Applied Basic Research Foundation under Grant 2022A1515110565. (\textit{Corresponding author: Xiaolong~Yang}.)}
\thanks{Lisha Shuai, Jiuling Dong, Shaofeng Tan, Haokun Zhang, Zilong Song, Gaoya Dong, and Xiaolong Yang are with the School of Computer and Communication Engineering, University of Science and Technology Beijing, Beijing, 100083, China (e-mail: promethusls@outlook.com; hndongjiuling@163.com; winber@126.com; zhanghaokun223@126.com; d202210376@xs.ustb.edu.cn; gaoyadong@ustb.edu.cn; yangxl@ustb.edu.cn).}
\thanks{Nan Zhang is with the State Grid Information \& Telecommunication Group Co., Ltd., Beijing, 101200, China (e-mail: 837885930@qq.com).}
\thanks{Shaofeng Tan is with the Information Management Center, Beijing Pinggu Hospital, Beijing, 101200, China.}
}

%
%

\markboth{Journal of \LaTeX\ Class Files,~Vol.~14, No.~8, August~2015}%
{Shell \MakeLowercase{\textit{et al.}}: Bare Demo of IEEEtran.cls for IEEE Journals}
%



\maketitle

\begin{abstract}
Local Differential Privacy (LDP) is a widely adopted privacy-protection model in the Internet of Things (IoT) due to its lightweight, decentralized, and scalable nature. However, it is vulnerable to poisoning attacks, and existing defenses either incur prohibitive resource overheads or rely on domain-specific prior knowledge, limiting their practical deployment. To address these limitations, we propose PEEL, a \underline{P}oisoning-\underline{E}xposing \underline{E}ncoding theoretical framework for \underline{L}DP, which departs from resource- or prior-dependent countermeasures and instead leverages the inherent structural consistency of LDP-perturbed data. As a non-intrusive post-processing module, PEEL amplifies stealthy poisoning effects by re-encoding LDP-perturbed data via sparsification, normalization, and low-rank projection, thereby revealing both output and rule poisoning attacks through structural inconsistencies in the reconstructed space. Theoretical analysis proves that PEEL, integrated with LDP, retains unbiasedness and statistical accuracy, while being robust to expose both output and rule poisoning attacks. Moreover, evaluation results show that LDP-integrated PEEL not only outperforms four state-of-the-art defenses in terms of poisoning exposure accuracy but also significantly reduces client-side computational costs, making it highly suitable for large-scale IoT deployments.
\end{abstract}


%
\IEEEpeerreviewmaketitle

\section{Introduction}
%
%
%
%

\IEEEPARstart{L}{ocal} Differential Privacy (LDP) is a rigorous privacy-preserving paradigm in the distributed setting, where each client perturbs its data through a lightweight randomizer, thereby protecting individual privacy without requiring other trusted third party, while still allowing meaningful statistical queries. Given these characteristics, LDP is increasingly deployed on the Internet of Things (IoT) edge devices to enable lightweight privacy-preserving data collection. Major tech firms leverage LDP for tasks such as gathering geolocation data (Microsoft \cite{ding2017collecting}, Xiaomi, Meizu), browsing habits (Google \cite{erlingsson2014rappor}), and emoji usage patterns from user input (Apple  \cite{2017LearningWP}). Beyond consumer applications, LDP is also critical in the IoT, including smart grids \cite{shanmugarasa2023local}, connected healthcare \cite{arachchige2019local}, and industrial control systems \cite{jiang2021differential}, among others. 

However, the randomizers of LDP, while essential for privacy, inadvertently enable data poisoning attacks by making it difficult for the aggregator to distinguish poisoned data from legitimate ones. Adversaries can deliberately inject poisoning to distort the Statistical Query Results (SQRs), compromising the validity of statistical analyses~\cite{wu2022poisoning}. Such attacks can bias mean/frequency estimation~\cite{cheu2022differentially,cao2021data}, skew histogram statistics~\cite{cheu2021manipulation}, degrade graph analytics~\cite{li2022fine}, and disrupt key-value aggregation~\cite{imola2022robustness}. In mission-critical applications, the corrupted SQRs can have far-reaching impacts, jeopardizing household safety, critical infrastructure, and societal stability~\cite{shuai2025poisoncatcher}.

Prior works have developed diverse countermeasures against LDP poisoning attacks, which can be broadly classified by their intervention stage. (i) \textbf{Pre-perturbation defenses} reinforce randomizers to reduce the feasibility of poisoning injection. Prior works included verifiable mechanisms~\cite{10.1007/978-3-030-81242-3_3}, collaborative protocols~\cite{10220122}, and hybrid cryptographic schemes~\cite{10.1007/978-3-031-68208-7_18}. However, these approaches often rely on prior assumptions about poisoning behaviors and introduce substantial overhead in communication and computation; (ii) \textbf{In-process mitigations} incorporate adaptive control strategies into either the client-side randomizer or the server-side aggregator, limiting the propagation and cumulative impact of poisoned inputs throughout the data collection pipeline. Representative techniques include anomaly-aware client filtering~\cite{279934} and dynamic perturbation reallocation~\cite{shuai_rdp_2022}. However, their reliance on continuous monitoring of perturbation outputs and frequent adaptive parameter updates results in substantial overhead, making deployment in high-throughput IoT environments impractical. (iii) \textbf{Post-hoc detections} reveal poisoning by detecting inconsistencies between SQRs and expected patterns. Representative designs include normalization with conditional checks~\cite{272214}, two-round anomaly detection~\cite{10415225}, bilevel optimization frameworks~\cite{10423870}, EM-based statistical defenses~\cite{app14146368}, and four-stage poison identification~\cite{shuai2025poisoncatcher}. These methods apply across different poisoning strategies, but their detection reliability depends on the accuracy of SQRs and the stability of the LDP randomizers.

These measures often suffer from prohibitive resource costs in IoT and other LDP environments, or a reliance on domain-specific prior knowledge, which collectively hinder their practical deployment. Therefore, we present PEEL, a Poisoning-Exposing Encoding theoretical framework for LDP, realized as a non-intrusive post-processing module that operates on LDP-perturbed data, and transcends the conventional classification of LDP defenses. Methodologically, we establish an architectural principle wherein any LDP mechanism capable of generating or being efficiently mapped to a 1-sparse vector can serve as the input layer within PEEL's unified framework. Within PEEL, we apply low-dimensional projection to 1-sparse encoded inputs followed by linear reconstruction. Benign samples maintain stable and consistent support structures throughout projection-reconstruction, enabling faithful recovery of their original sparse patterns. In contrast, poisoning attacks disrupt this sparse geometry, inducing support misalignment and instability that manifest as statistically distinctive residuals. These residual signatures thus establish a robust, attack-agnostic criterion for poisoning exposure. The main contributions are as follows:
\begin{itemize}
    \item We propose PEEL, a lightweight theoretical framework for LDP that exposes poisoning attacks by verifying structural consistency, eliminating the need for domain-specific priors.
    \item We establish rigorous theoretical guarantees for PEEL, demonstrating its ability to ensure unbiased estimation while preserving the original statistical utility bounds. Furthermore, we provide a comprehensive robustness analysis against both output and rule poisoning attacks.
    \item Extensive evaluations demonstrate that PEEL achieves significantly higher poisoning detection accuracy over four state-of-the-art methods, while consistently incurring lower client-side overhead across seven mainstream defense approaches, confirming its practical advantages for large-scale IoT deployments.
\end{itemize}

\section{Related Works}
\label{sec:related_work}

\textbf{LDP in single-attribute data collection.} As the foundational use case of LDP, single-attribute data collection has been extensively studied and widely deployed. Prior research in this setting generally falls into three major categories that focus on frequency estimation and heavy-hitter identification for categorical data, mean estimation for continuous data, and mechanism design that targets structural optimization. 

For frequency estimation, existing works have developed techniques that efficiently reconstruct categorical distributions under privacy constraints. Representative methods span randomized response~\cite{warner1965randomized}, Bloom-filter ~\cite{erlingsson2014rappor}, one-hot encoding with count sketch~\cite{2017LearningWP}, Hadamard transforms~\cite{acharya2019hadamard}, optimized hashing schemes~\cite{wang2017locally, kairouz2016discrete}, prefix-tree-based encoding~\cite{wang2019locally}, and projection-based transforms~\cite{bassily2015local, joseph2018local}. For mean estimation of continuous data, LDP mechanisms have evolved from early noise-injection approaches~\cite{dwork2006calibrating} to more refined, information-theoretically grounded methods. These include geometric encoding~\cite{duchi2013local, duchi2014privacy}, symbolic quantization with low-bit perturbation~\cite{ding2017collecting}, and minimax-optimal schemes based on Gaussian noise or sign compression~\cite{lee2023minimax,nikita2025efficient}. Beyond task-specific designs, studies have investigated the theoretical limits of general-purpose mechanisms under LDP. Core contributions include extremal mechanism ~\cite{kairouz2014extremal}, staircase-structured perturbation~\cite{geng2015staircase}, trusted-party assisted protocols ~\cite{ren2018textsf}, and unified key-value estimation frameworks ~\cite{ye2021privkvm}.

\textbf{LDP in multi-attribute data collection.} Beyond single-attribute collection, a substantial body of work has examined multi-attribute data, where high dimensionality raises unique challenges such as privacy budget allocation, dimensionality curse, and cumulative noise amplification \cite{8911204}. These issues substantially impair statistical accuracy and limit the scalability of LDP in multi-attribute applications.

Prior works have pursued three major technical directions. The first centers on accuracy-enhancing mechanisms that operate directly on the original data domain using optimized perturbation encodings without relying on complex reconstruction. Representative methods include marginal encoding with Hadamard transforms~\cite{cormode2018marginal}, structured decomposition~\cite{yang2020answering}, joint aggregation protocols~\cite{kikuchi2022castell}, and sparse perturbation mechanisms~\cite{xu2023mlpkv, wang2025ukvldp}. The second emphasizes reconstruction through iterative optimization frameworks, extracting latent statistical structures from noisy inputs to recover accurate aggregates. Notable techniques include hierarchical decomposition with interaction queries~\cite{wang2019answering} and sparse signal recovery via iterative hard thresholding~\cite{wang2019sparse}. The third focuses on dimensionality reduction through feature selection to reduce system overhead. Approaches include randomized projection and 1-bit encoding~\cite{nguyen2016collecting} and adaptive partitioning with selective reporting~\cite{wang2019collecting}, while theoretical advances have established minimax-optimal noise allocation strategies~\cite{duchi2018minimax}.

Despite differences in randomizers, both single- and multi-attribute data collection mechanisms commonly generate outputs that are structurally constrained. For example, one-hot and Bloom-filter encodings activate only a few bits in each output, hash- or orthogonal-code mechanisms map inputs to fixed codewords with predetermined supports, and staircase or sign-based quantizers emit low-dimensional signed vectors. These perturbation rules enforce strong regularities in the output space, that is, each output is confined to a limited set of coordinates or templates, leading to highly structured patterns. Such structural uniformity not only streamlines aggregation and stabilizes statistical estimators but also reveals that LDP outputs are far from arbitrary, namely, they occupy a constrained subspace shaped by the mechanism’s design. This observation highlights an often-overlooked property, i.e., the privacy guarantee operates within a rigid encoding structure, which serves as a fundamental lever for both attack design and defense development in LDP systems.

\section{Preliminaries}

\subsection{Local Differential Privacy}

This study adopts the standard $\varepsilon$-LDP definition, applicable to both single- and multi-attribute data collection settings. Let $X_i \in \mathcal{X}$ be the input and $Z_i \in \mathcal{Z}$ be the output corresponding to the LDP on $\mathcal{X}$, where $\mathcal{X}$ denotes the domain of raw data values, $\mathcal{Z}$ denotes the output space of the local mechanism. For any $x, x' \in \mathcal{X}$ and any measurable subset $S \subseteq \mathcal{Z}$, given a privacy budget $\varepsilon > 0$, if a statistical query function $\mathcal{Q}$ satisfies the following inequality, then $Z_i$ is said to be an $\varepsilon$-LDP representation of $X_i$ \cite{duchi2013local}:
\begin{equation}
    \sup_{S \subseteq \mathcal{Z}} \frac{\mathcal{Q}(Z_{i} \in S | X_{i} = x)}{\mathcal{Q}(Z_{i} \in S | X_{i} = x')} \leq e^\varepsilon
    \label{LDP_definition_1}
\end{equation}
where $\mathcal{Q}(Z_i \in S \mid X_i = x)$ denotes the conditional probability that the output $Z_i$ falls within the set $S$, given the input $X_i = x$, denoted as $\mathbb{P}(Z_i \in S \mid X_i = x)$. Each output $Z_i$ depends solely on its corresponding input $X_i$, represented as $X_i \rightarrow Z_i$, and $Z_i$ is independent of all other inputs and outputs given $X_i$, expressed as $Z_i \perp \{X_j, Z_j \mid j \neq i\} \mid X_i$.

\subsection{LDP Poisoning Attacks}

The randomizer in LDP acts as a natural layer of obfuscation, rendering poisoned outputs indistinguishable from legitimate ones and hindering detection. To characterize poisoning threats in the LDP setting, we build upon the attack framework developed in our prior work~\cite{shuai2025poisoncatcher}, which categorizes attacks according to their manipulation targets within the perturbation process. Specifically, three representative classes are distinguished: (i) input poisoning, which manipulates client inputs before perturbation; (ii) output poisoning, which tampers with perturbed outputs after perturbation; and (iii) rule poisoning, which alters the internal parameters of the randomizers. 

In input poisoning attacks, adversaries compromise the quality of raw data collected at the client side, aiming to disrupt the integrity of downstream statistical queries. These attacks commonly follow two primary strategies, where adversaries either impersonate legitimate clients to inject crafted inputs or manipulate the sensing environment to induce corrupted readings. As such attacks require no access to centralized infrastructure or elevated system privileges, they are inherently low-cost, stealthy, and difficult to detect in distributed settings.

In output poisoning attacks, adversaries manipulate perturbed reports after the LDP perturbation step. These attacks commonly follow two strategies: adversaries may either alter perturbed reports before submission or inject forged outputs that circumvent the legitimate reporting channel. Because perturbation and transmission are decoupled, adversaries can decide whether to tamper with individual client outputs or inject bulk forgeries at the aggregator, thus flexibly controlling both the injection point and the attack scale. Formally, the output poisoning attack is modeled as a post-processing map within a plausible set $\mathcal{X}\subseteq \mathrm{Range}(\psi_\varepsilon)$:
\begin{equation}
\Delta_{\mathrm{out}}(x,\psi_\varepsilon)\ \coloneqq\ \Delta\!\big(\psi_\varepsilon(x)\big),
\label{output_poisoning_attacks}
\end{equation}
where $\psi_\varepsilon$ is an $\varepsilon$-LDP randomization mechanism, $\Delta$ denotes the poisoning manipulation, and the probability of the poisoning mechanism outputting a particular value $x$ is proportional to $\exp\!\left(-\frac{\varepsilon\,\lVert x-\psi_\varepsilon(x)\rVert_1}{f}\right)$, with $\psi_\varepsilon$ being the intended LDP mechanism and $f$ the query sensitivity.

In rule poisoning attacks, the adversary rewrites the local randomizer (e.g., encoding logic or privacy parameters), distorting the benign input–output mapping while ensuring that audit-visible accounting remains unchanged. Formally, the intended $\psi_\varepsilon$ is replaced by a modified mechanism $\Delta(\psi_\varepsilon)$, applied to each report data $x$:
\begin{equation}
    \Delta_{\mathrm{rule}}(x,\psi_\varepsilon) \coloneqq \Delta(\psi_\varepsilon)(x),
\text{ s.t.}\ \sum_{i=1}^{n}\varepsilon_i=\varepsilon_{\text{total}}.\label{rule_poisoning_attacks}
\end{equation}

Because the reported privacy budgets are preserved, these small but systematic deviations persist across all data and accumulate in aggregation, making rule poisoning both stealthy to audits and damaging to SQRs.

\subsection{Problem Definition}

While all three classes of poisoning attacks bias SQRs, they differ in execution and impact. Input poisoning occurs during local data acquisition, compromising data integrity at the source. Such attacks are typically mitigated via local verification or outlier filtering, which \textbf{is outside the scope of this work}. Output and rule poisoning attacks manipulate either the perturbed reports or the LDP internal logic, while the outputs maintain apparent structural compliance with benign patterns, thereby evading conventional defenses and undermining estimation fidelity. \textbf{This work primarily focuses on these stealthy poisoning attacks}.

Although they intervene at different points, output and rule poisoning both act on perturbed outputs and are modeled in a unified manner by denoting any poisoned output as $z_i^\Delta$. Formally,
\begin{equation}
    z_i^\Delta = \Delta(z_i), 
\end{equation}
where $\Delta:\mathcal{Z}\to\mathcal{Z}$ preserves the output domain $\mathcal{Z}$ and is not required to satisfy $\varepsilon$-LDP. Hence, poisoned outputs remain elements of $\mathcal{Z}$ and appear compliant while embedding targeted deviations.

Based on properties of representative LDP mechanisms summarized in Section~\ref{sec:related_work}, the benign perturbed output $z_i$ can be expressed as:
\begin{equation}
z_i = \psi_{\varepsilon}(x_i) = \mathcal{C}(x_i) + \mathcal{R}_i^{(\varepsilon)} ,
\label{Structured_coding_random_perturbation}
\end{equation}
where $\mathbb{P} \left(\left\|\mathcal{R}_i^{(\varepsilon)}\right\| \le \delta_\varepsilon\right) \ge 1-\eta_\varepsilon$, $\mathcal{C}(\cdot)$ denotes a structure-preserving encoding of the raw input $x_i$, and $\mathcal{R}_i^{(\varepsilon)}$ is an $\varepsilon$-LDP-compliant randomization term. The norm $\|\cdot\|$ is defined on the output space $\mathcal{Z}$ and is mechanism-specified (e.g., $\ell_2$ or $\ell_\infty$). Constants $\delta_\varepsilon$ and $\eta_\varepsilon$ are mechanism- and $\varepsilon$-dependent.

The perturbed outputs are expected to satisfy the structural consistency bound:
\begin{equation}
\| z_i - \mathcal{C}(x_i) \| \le \delta_\varepsilon ,
\end{equation}
which follows from the calibration of the randomized mechanism. This holds with probability at least $1-\eta_\varepsilon$ by calibration of the randomized mechanism.

Let $\mathcal{T}_{\mathrm{struct}}$ denote an encoding operator that renders violations of the expected support constraints detectable. Define $\mathcal{G}$ as the admissible set in the representation space induced by benign outputs:
\begin{equation}
    \mathcal{G} \!=\! \left\{ \mathcal{T}_{\mathrm{struct}}(z) \,\middle|\, z=\mathcal{C}(x)+\mathcal{R}^{(\varepsilon)},\ \|\mathcal{R}^{(\varepsilon)}\|\le \delta_\varepsilon \right\}.
\end{equation}

Poisoning exposure occurs when the transformed output lies outside $\mathcal{G}$:
\begin{align}
&\mathcal{S}_{\mathrm{poison}} =
\left\{ z_i^\Delta \middle|
\mathcal{T}_{\mathrm{struct}}(z_i^\Delta) \notin \mathcal{G} \right\}.
\label{eq:poisoning-exposure-struct}
\end{align}

The violation of this admissible set constraint serves as a robust indicator of data poisoning.

\section{PEEL: Poisoning-Exposing Encoding Mechanism for LDP}

As a concrete instantiation of the transformation function $\mathcal{T}_{\mathrm{struct}}$, PEEL is a structure-oriented encoding mechanism for LDP that maps perturbed outputs into a representation space where poisoning-induced structural deviations are amplified into explicit forms. The process comprises sparse mapping, normalization, and low-rank projection, jointly preserving the structural patterns characterized by benign data under legitimate $\varepsilon$-LDP perturbations as defined in $\mathcal{G}$. Because the encoding pipeline determines these patterns, any poisoning-induced modification disrupts their alignment with the expected structural support in this space, making such deviations directly observable. 

Fig. \ref{fig:peel-flow} illustrates the data flow of PEEL under LDP, tracing the complete pathway from raw data on the client side through LDP perturbation and PEEL encoding, to PEEL decoding, and final statistical query on the receiver side. Within this framework, the PEEL-encoded vector $y$ serves as a transmission object specifically engineered for efficient communication and robust defense against output and rule poisoning attacks. Conversely, the decoded vector on the receiver side represents the 1-sparse data with standard LDP semantics. All subsequent utility analyses are performed on $s$, as it directly corresponds to the input for statistical query tasks.

\begin{figure}[htbp]
  \centering
  \includegraphics[width=0.95\linewidth]{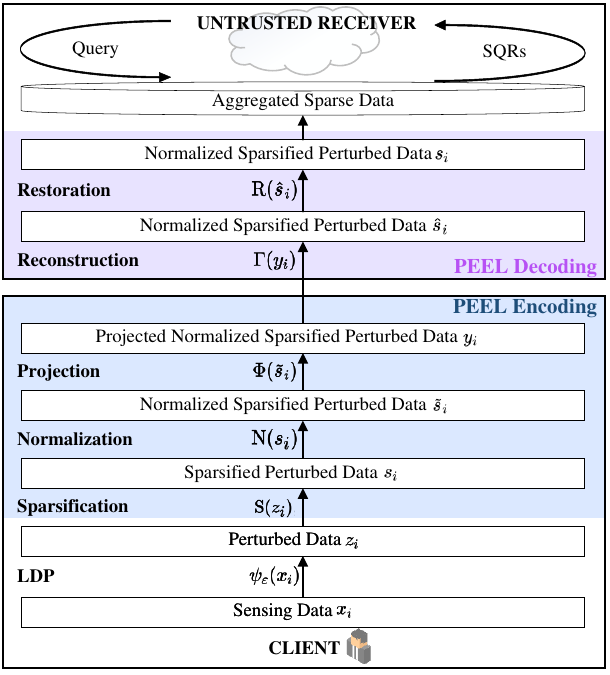}
  \caption{Data Flow of LDP-integrated PEEL}
  \label{fig:peel-flow}
\end{figure}

\vspace{-2ex}

\subsection{Mathematical Model}

The PEEL encoding process transforms LDP-perturbed outputs into a space where poisoning-induced deviations become structurally explicit, through three sequential stages: (i) sparse mapping enforces a $1$-sparsity form to localize deviations in non-sparse outputs or align already 1-sparse outputs to a predefined structural position, ensuring they share the same sparsity pattern for direct comparison; (ii) normalization standardizes the sparse representation onto a mechanism-consistent scale, ensuring comparability across samples and mechanisms; and (iii) low-rank projection maps the standardized sparse representation into a low-dimensional subspace in which legitimate encodings can be exactly reconstructed, whereas any tampering yields a residual that exposes poisoning.

\noindent\textbf{Sparse Mapping.} Let $\mathrm{S}:\mathcal{Z}\to\mathcal{Z}_\mathrm{S}$ denote the sparse mapping function applied to each perturbed output $z_i$. The imposed $1$-sparsity constraint localizes any structural deviation to a single coordinate, concentrating the effect of a poisoning injection rather than dispersing it across dimensions. 

This extreme sparsification is chosen for its twofold advantage. First, it maximizes detection sensitivity by simplifying legitimate patterns, which makes deviations more pronounced and reduces the risk of false positives or missed detections. Second, it ensures computational efficiency, which is crucial for resource-constrained IoT edge devices. Formally,
\begin{equation}
    s_i = \mathrm{S}(z_i), \quad s_i \in \mathbb{R}^k,\quad \|s_i\|_0 \le 1,
    \label{Sparse_Mapping}
\end{equation}
where $k$ denotes the data dimension, and the non-zero entry retains the sign of the corresponding component in $z_i$, preserving a deterministic structure.

i) \textbf{Naturally 1-sparse mechanisms.} For LDP randomizers whose client side data are already 1-sparse with symmetric signs and known selection probabilities (e.g., RR \cite{warner1965randomized} / kRR \cite{pmlr-v48-kairouz16} / Direct Encoding / LH / OLH \cite{wang2017locally}, Hadamard Response family \cite{cormode2018marginal}\cite{yang2020answering}, and Harmony \cite{nguyen2016collecting} / Duchi-style 1-bit \cite{duchi2018minimax} / Piecewise Mechanisms \cite{wang2019collecting}, as well as the k-Subset case with $m=1$) \cite{geng2015staircase}, as a structurally consistent class characterized by single index and symmetric sign and known selection probabilities. Client reports in this class are intrinsically 1-sparse (or become 1-sparse under a fixed linear recoding), and therefore require no additional sparsification. In this case, $s_i=z_i$ in \eqref{Sparse_Mapping}.

ii) \textbf{Non–1-sparse mechanisms.} Unary Encoding (UE) / Optimized UE (OUE)\cite{wang2017locally}, RAPPOR\cite{erlingsson2014rappor}, $k$-Subset with $m{>}1$ / Subset Selection\cite{geng2015staircase}, additive-noise numeric mechanisms\cite{dwork2006calibrating}, and spherical-direction reports all produce multi-dimensional or dense outputs\cite{bhowmick2018protection}. For these LDP mechanisms, we apply a sparsification map $S$ and enforce the conditional–expectation alignment, thereby casting heterogeneous data into a unified 1-sparse normal form without altering unbiasedness.

Let $t(z_i)\in\mathbb{R}^k$ denote the per-coordinate unbiased transform used by the LDP statistical estimator. We require the sparse mapping $s_i=\mathrm{S}(z_i)$ to satisfy the conditional expectation alignment, i.e.,
\begin{equation}
\mathbb{E}\big[s_i \,\big|\, z_i\big] = t(z_i).
\label{eq:HT-alignment}
\end{equation}

The above condition is sufficient to ensure that, for any linear or dimension-wise aggregation query $Q$, the expected value under PEEL matches that of the standard LDP pipeline, introducing no additional bias:
\begin{equation}
\mathbb{E}\big[Q(s_{1:n})\big]
=
\mathbb{E}\big[Q\big(t(z_{1:n})\big)\big].
\end{equation}
where the subscript $1\!:\!n$ denotes the sequence of data points from the first to the $n$-th sample.

1-sparsification method that satisfies \eqref{eq:HT-alignment} is an unequal-probability sampling construction with inverse-probability weighting. Choose selection probabilities $p_j(z_i)\in(0,1]$ such that $\sum_{j=1}^k p_j(z_i)=1$ and $p_j(z_i)>0$ whenever $t_j(z_i)\neq 0$. Draw a single index $J\sim p(z_i)$ and define the 1-sparse output:
\begin{equation}
s_{i,J}=\frac{t_J(z_i)}{p_J(z_i)},
\label{eq:ht-construction}
\end{equation}
where $s_{i,j}=0$ and $(j\neq J)$. Then for any coordinate $j$,
\begin{equation}
\!\!\!\mathbb{E}\big[s_{i,j}\,\big|\,z_i\big]
\!=\!\frac{t_j(z_i)}{p_j(z_i)}\,\mathbb{P} \big(J\!=\!j\,\!\big |z_i\big)
\!=\!\frac{t_j(z_i)}{p_j(z_i)}\,p_j(z_i)
\!=\! t_j(z_i),\!\!\!
\end{equation}
so \eqref{eq:HT-alignment} holds by construction.

All operations in $\mathrm{S}$ depend only on $z_i$ and on auxiliary randomness that is independent of the raw data $x_i$, and therefore constitute post-processing that does not degrade the original $\varepsilon$-LDP privacy guarantee.

\noindent\textbf{Normalization.} Let $\mathrm{N}:\mathcal{Z}_\mathrm{S} \to \mathcal{Z}_\mathrm{N}$ denote the normalization operator applied to the sparse structural vector $s_i$.  This step mitigates inconsistencies in numerical scales arising from different LDP mechanisms or feature magnitudes by applying z-score normalization to $s_i$, which preserves the sign of each non-zero entry while rescaling its magnitude to a common scale for cross-sample comparability. Formally,
\begin{equation}
    \tilde{s}_i = \mathrm{N}(s_i).
\end{equation}

The normalization is computed solely from that data's own coordinates—using its within-vector mean and standard deviation—without referencing other data or the raw sensitive input. Consequently, it is a post-processing step on the LDP output, aligned with the front-end LDP workflow and incurring no additional privacy cost.

Let the original 1-sparse vector be $s=\pm e_J$, where $e_J$ denotes the $J$-th standard basis vector. The standardized vector $\tilde{s}$ then satisfies:
\begin{equation}
\tilde{s}_J=\pm\sqrt{k-1},\;
\tilde{s}_j=\mp\frac{1}{\sqrt{k-1}}\quad (j\neq J).
\label{z-score_data}
\end{equation}

\noindent\textbf{Low-Rank Projection.} Let $\mathrm{P}:\mathcal{Z}_\mathrm{N}\!\to\!\mathbb{R}^{k-1}$ project the normalized one-sparse code $\tilde{s}_i$ using a data-independent Gaussian map $\Phi\in\mathbb{R}^{(k-1)\times k}$. Writing $\Theta=\Phi W$ with $W\in\mathbb{R}^{k\times (k-1)}$ an orthonormal structural basis, $\Theta$ is square and invertible with probability $1$, and is a subspace near-isometry on $\mathrm{col}(W)$. For any benign $\tilde{s}\in\mathrm{col}(W)$,
\begin{equation}
    (1-\varepsilon)\,\|\tilde{s}\|_2 \le \|\Phi\tilde{s}\|_2 \le (1+\varepsilon)\,\|\tilde{s}\|_2,
\end{equation}
with high probability. Consequently, relative geometry among benign encodings is preserved in $\mathbb{R}^{k-1}$, while poisoning-induced deviations become more salient in the compact projected space.

Each position in $s_i$ can represent two sign-symmetric states (positive or negative), resulting in $2k$ admissible $1$-sparse normalized encodings in total. These admissible encodings constitute the columns of the canonical structural matrix $\mathcal{D} \in \mathbb{R}^{k \times 2k}$, which is mean-centered and symmetric, introducing linear dependencies between columns and limiting its rank to at most $k-1$.

The minimal subspace that spans all legitimate encodings is obtained by solving the following constrained optimization problem:
\begin{equation}
    \min_{W, A} \left\| \mathcal{D} - WA \right\|_F^2, 
    \quad \text{s.t. } W^\top W = I,
\end{equation}
where $W \in \mathbb{R}^{k \times (k{-}1)}$ is a column-orthonormal basis and $A \in \mathbb{R}^{(k{-}1) \times 2k}$ are projection coefficients. This corresponds to finding the optimal low-rank representation of $\mathcal{D}$ in the least-squares sense. 

Given $W$ from the decomposition, the low-rank projection of $\tilde{s}_i$ is computed as:
\begin{equation}
    \alpha_i = W^\top \tilde{s}_i, \label{s_i_low_rank_projection},
\end{equation}
and the approximate reconstruction within this subspace is:
\begin{equation}
    \hat{s}_i \approx W \alpha_i + e_i,
\end{equation}
where $e_i$ denotes the reconstruction residual. 

\noindent\textbf{Structural Encoding.} Let $\mathcal{T}_{\mathrm{encode}}: \mathcal{Z} \to \mathbb{R}^{k-1}$ represent PEEL's structural encoding, comprising sparse mapping, normalization, and low-rank projection in sequence. Formally,
\begin{equation}
    y_i = \mathcal{T}_{\mathrm{encode}}(z_i) 
        = \Phi \cdot \mathrm{N}(\mathrm{S}(z_i)),
    \label{Structural_consistency_expression}
\end{equation}
where $y_i$ serves as a unified structural encoding for all perturbed samples, enabling the aggregation server to evaluate structural consistency across inputs.

\noindent\textbf{Reconstruct and Consistency Exposure.} Let $\mathcal{T}_{\text{decode}}: \mathbb{R}^{k-1} \to \mathbb{R}^k$ denote PEEL's structural decoding, comprising inverse low-rank projection.  Together with the encoding operator $\mathcal{T}_{\text{encode}}$, they form the complete structural transformation $\mathcal{T}_{\text{struct}}$ of PEEL. To verify consistency, the aggregator applies a linear consistency reconstruction operator $\Gamma:\mathbb{R}^{k-1} \to \mathbb{R}^k$ to recover an estimate $\hat{s}_i$ of the normalized 1-sparse representation, i.e.,
\begin{equation}
    \hat{s}_i = \Gamma y_i.
    \label{Projective_Reconstruction}
\end{equation}

Under structural consistency, the reconstruction is exact because the operator $\Gamma$ serves as the left-inverse of the projection matrix $\Phi$ on the subspace of legitimate encodings. Specifically, with \eqref{Structural_consistency_expression} and $\Gamma \in \mathbb{R}^{k \times (k{-}1)}$ chosen as the Moore--Penrose pseudoinverse (or equivalently $\Phi^\top$ under RIP), \eqref{Projective_Reconstruction} satisfies:
\begin{equation}
    \hat{s}_i = \Gamma y_i = \Gamma \Phi \tilde{s}_i = \tilde{s}_i. \label{reconstruction_process}
\end{equation}

For valid inputs, the mapping $\Phi$ and reconstruction operator $\Gamma$ are lossless over the subspace spanned by the $2k$ admissible $1$-sparse normalized encodings in $\mathcal{D}$.  Thus, $\hat{s}_i$ exactly matches one of these discrete patterns, each characterized by a single dominant coordinate and $(k-1)$ suppressed coordinates. In contrast, poisoning injections alter either the dominant coordinate’s position or its relative magnitude, moving the perturbed representation outside this admissible subspace  (where $\mathcal{G} = \mathcal{D}$ in \eqref{eq:poisoning-exposure-struct}). The resulting reconstruction $\hat{s}_i$ exhibits a non-zero residual and deviates from all legitimate patterns, yielding values inconsistent with benign encodings and thereby exposing the poisoning.

\noindent\textbf{Structural Restore.} This step restores the canonical 1-sparse representation in the reconstruction space, thereby enabling statistical queries on the receiver side to operate on a unified representation.

For $\hat{s}_i$, define the deterministic restore operator $\mathrm R$:
\begin{equation}
\mathrm R(\hat{s}_i) \coloneqq \operatorname{sgn}\!\big(\hat{s}_{i,J}\big)\, e_J,
\text{ where } 
J \coloneqq \operatorname*{arg\,max}_{j\in\{1,\dots,k\}} \lvert \hat{s}_{i,j}\rvert .
\label{restore}
\end{equation}

Under the single-data z-score, the unique maximum-amplitude coordinate identifies the support and its sign. In the closed-loop setting $\hat{s}_i=\tilde{s}_i$, this yields:
\begin{equation}
   \mathrm R(\hat{s}_i)=s_i.
\end{equation}

For LDP randomizers that are inherently 1-sparse with symmetric signs and known selection probabilities, the above restoration is lossless. In contrast, LDP mechanisms with multi-dimensional or dense reports are first mapped to a 1-sparse surrogate via sampling with inverse-probability weighting, ensuring the alignment \eqref{eq:HT-alignment}. Consequently, any linear or dimension-wise statistical query attains the same expectation as in the standard LDP pipeline, preserving unbiasedness while providing a unified representation for downstream analysis.

\noindent\textbf{Closed-Loop PEEL Process.} The encode–decode pathway of PEEL consists of sparse mapping $\mathrm{S}$, normalization $\mathrm{N}$, low-rank projection $\Phi$, linear reconstruction $\Gamma$, and restore operator $\mathrm R$.  This closed-loop process preserves structurally valid inputs exactly after reconstruction, enabling consistency verification in the reconstructed space. Formally, the transformation pathway for a perturbed sample $z_i$ is:
\begin{equation}
    z_i \xrightarrow{\mathrm{S}} s_i 
        \xrightarrow{\mathrm{N}} \tilde{s}_i 
        \xrightarrow{\Phi} y_i 
        \xrightarrow{\Gamma} \hat{s}_i  = \tilde{s}_i
        \xrightarrow{\mathrm{R}} s_i ,
    \label{eq:closed_loop}
\end{equation}

This closed-loop property ensures that benign inputs are reconstructed without distortion, whereas poisoning-induced deviations result in observable reconstruction residuals, forming the basis for structure-oriented poisoning exposure.

\subsection{Poisoning-Exposing Principle}

Whereas prior defenses rely on distributional similarity or strong attack-model assumptions to separate benign from poisoned data, PEEL shifts the detection paradigm by leveraging structural consistency. Through the projection–reconstruction process, even small inconsistencies in encoding structure translate into large reconstruction residuals~\cite{mitra2012analysis, studer2011recovery}. The detection-mode shift and its amplification effect make poisoned samples noticeable beyond the reach of conventional statistical defenses.

\subsubsection{Structural Reversibility \label{Theoretical_guarantee_of_reversibility}}

Let $\tilde{\mathcal{D}} \in \mathbb{R}^{k \times 2k}$ denote the standardized structural matrix, where each column corresponds to one admissible $1$-sparse normalized encoding. Since there are $2k$ such encodings (two sign-symmetric states per axis), $\tilde{\mathcal{D}}$ collects them into a single canonical representation. PEEL performs principal direction decomposition to extract a column-orthogonal basis $W \in \mathbb{R}^{k \times (k{-}1)}$, whose column space $\mathrm{col}(W)$ defines the structural subspace. A composite projection matrix is then defined as:
\begin{equation}
    \Theta = \Phi W,\label{composite_projection_matrix}
\end{equation}
Where $\Theta \in \mathbb{R}^{(k{-}1) \times (k{-}1)}$. Given a normalized sparse structural vector $\tilde{s}_i \in \mathrm{col}(W)$ derived from a perturbed sample $z_i$, its projected representation can be expressed as:
\begin{equation}
    y_i = \Phi \tilde{s}_i = \Theta \alpha_i,
\end{equation}
where $\alpha_i$ denotes the projection coefficients in the structural basis, see \eqref{s_i_low_rank_projection}.

When $\Theta$ is a full-rank square matrix, its Moore-Penrose pseudoinverse $\Theta^\dagger$ coincides with its true inverse $\Theta^{-1}$. The inverse transformation is thus defined by:
\begin{equation}
    \Gamma = W \Theta^{-1} = W (\Phi W)^{-1},\label{inverse_transformation}
\end{equation}
where $\Gamma \in \mathbb{R}^{k \times (k{-}1)}$. This yields an exact reconstruction path,\eqref{reconstruction_process} satisfies:
\begin{equation}
    \hat{s}_i = \Gamma y_i = W \Theta^{-1} y_i = \tilde{s}_i. \label{eq:perfect_reconstruction}
\end{equation}

Although the proposed linear transformation path in PEEL resembles sparse coding and projection steps commonly used in compressive sensing, the underlying mechanism fundamentally differs from that paradigm: it is deterministic, full-rank, and analytically invertible. Specifically, the normalized structural vectors $\tilde{s}_i$ are explicitly constrained to lie within the subspace $\mathrm{col}(W)$, and the composite projection matrix $\Theta$ is full-rank by construction. As a result, the mapping is not an underdetermined recovery problem, but a deterministic and analytically invertible linear transformation. Consequently, the reconstruction path $\hat{s}_i = \tilde{s}_i$ achieves exact recovery without approximation or sparsity-driven inference.

The exact reconstruction relation in \eqref{eq:perfect_reconstruction} ensures lossless recovery of legitimate encodings under the following conditions:
\begin{itemize}
    \item[(C1)] \textbf{Single-active structural component:} $\tilde{s}_i$ is 1-sparse.
    \item[(C2)] \textbf{Subspace alignment:} $\tilde{s}_i \in \mathrm{col}(W)$.
    \item[(C3)] \textbf{Full-rank projection:} $\Theta$ is a nonsingular square matrix.
\end{itemize}

Since $\tilde{\mathcal{D}}$ is composed of symmetric 1-sparse normalized encodings, all its columns reside in $\mathrm{col}(W)$. Moreover, since $\Phi$ is drawn from a sub-Gaussian distribution, $\Theta$ (see \eqref{composite_projection_matrix}) is full-rank with probability $1$, ensuring that the inverse transformation $\Gamma$ (see \eqref{inverse_transformation}) exists and is closed-form. Therefore, PEEL achieves perfect reconstruction for any legitimate sample that satisfies conditions (C1)–(C3).

If $\Theta$ is approximately full-rank or $\tilde{s}_i$ slightly deviates from $\mathrm{col}(W)$, reconstruction becomes approximate but remains stable. The use of the Moore-Penrose pseudoinverse yields a least-squares optimal solution:
\begin{equation}
    \hat{s}_i = W \Theta^\dagger y_i = \tilde{s}_i + e_i, 
\end{equation}
where $e_i$ is the reconstruction deviation, and $\|e_i\|_2 \to 0$.

To facilitate theoretical guarantees on reconstruction fidelity and poisoning exposure, we introduce the following structural assumptions:
\begin{itemize}
    \item[(A1)] \textbf{Sparsity of perturbation structure:} Each structural code $s_i$ is 1-sparse. This enables unambiguous encoding and low-rank basis construction.    
    \item[(A2)] \textbf{Stable subspace decomposition:} The normalized structural matrix $\tilde{\mathcal{D}}$ spans a $(k{-}1)$-dimensional subspace with linearly independent columns, ensuring a well-defined projection space.
    \item[(A3)] \textbf{Spectral stability of projection:} The projection matrix $\Phi$ has full row rank, and its singular values are bounded away from zero, ensuring numerical stability and invertibility of the encoding path.
\end{itemize}

These properties define an idealized setting in which PEEL guarantees exact reconstruction of benign encodings, thereby establishing a reversible encoding map. This theoretical foundation enables precise assessment of structural consistency, which in turn facilitates the exposure of poisoning manipulations. In practice, however, practical LDP mechanisms often deviate from (A1)–(A3). Importantly, PEEL does not require strict satisfaction of these properties. Minor deviations only affect the reconstruction path of benign outputs within small tolerances, whereas poisoned outputs induce disproportionally larger residuals due to their structural misalignment. Thus, poisoning exposure remains effective even when the properties are relaxed.

\subsubsection{Exposure Mechanism}

Under conditions (C1)–(C3) established, PEEL enables lossless reconstruction for legitimate samples through a stable encoding–projection–reconstruction chain in a $(k{-}1)$-dimensional subspace. For poisoned samples, the same reconstruction path amplifies structural inconsistencies, yielding non-zero residuals that expose their deviation and render them geometrically distinguishable.

For benign perturbations, the normalized structural encoding $\tilde{s}_i$ satisfies conditions (C1)–(C3). As a result, PEEL ensures exact recovery of the encoding through $\hat{s}_i = \tilde{s}_i$, yielding zero reconstruction error $e_i = 0$ and establishing a closed-form benchmark for structural consistency.

However, in \textbf{output poisoning attacks}, although the poisoned encoding $s_i^\Delta$ retains the 1-sparsity property, its nonzero entry may not correspond to any legitimate column in $\tilde{\mathcal{D}}$. Consequently, the normalized encoding $\tilde{s}_i^\Delta$ falls outside $\mathrm{col}(W)$, violating the subspace alignment condition (C2). The resulting representation cannot be losslessly inverted. The aggregator reconstructs:
\begin{equation}
    \hat{s}_i^\Delta = \Gamma y_i^\Delta = \Gamma \Phi \tilde{s}_i^\Delta,
\end{equation}
and the reconstruction residual is defined as:
\begin{equation}
    e_i^\Delta = \| \hat{s}_i^\Delta - \tilde{s}_i^\Delta \|_2.
\end{equation}

Since $\tilde{s}_i^\Delta \notin \mathrm{col}(W)$, it follows that $\hat{s}_i^\Delta \in \mathrm{col}(W)$ but $\hat{s}_i^\Delta \neq \tilde{s}_i^\Delta$, ensuring $e_i^\Delta > 0$. This nonzero residual establishes the formal criterion by which PEEL exposes structural inconsistency induced by poisoning.

Similarly, in \textbf{rule poisoning attacks}, the perturbed structure $s_i^\Delta$ may satisfy the 1-sparsity constraint, yet its nonzero component is often manually injected by the adversary and does not align with any valid structural vector in the reference matrix $\tilde{\mathcal{D}}$. As a result, the normalized encoding $\tilde{s}_i^\Delta$ deviates from the subspace $\mathrm{col}(W)$ spanned by legitimate structural directions. This misalignment breaks the integrity of PEEL’s structural transformation pipeline. 

Specifically, the principal coordinate $\alpha_i^\Delta = W^\top \tilde{s}_i^\Delta$ derived from a poisoned sample lacks semantic validity, as $\tilde{s}_i^\Delta$ no longer aligns with any column in the structural basis $\tilde{\mathcal{D}}$. Its projected representation $y_i^\Delta = \Phi \tilde{s}_i^\Delta$ therefore falls outside the manifold formed by legitimate encodings. Consequently, the inverse mapping $\hat{s}_i^\Delta = \Gamma y_i^\Delta$ fails to recover $\tilde{s}_i^\Delta$, producing residuals that cannot be reconciled within the structure-consistent space. These unrecoverable deviations are amplified through the closed-loop reconstruction path, and the resulting residuals serve as stable, geometrically separable indicators of structural inconsistency, explicitly exposing rule poisoning behaviors in the encoded output space.

To generalize the analysis across poisoning types, we define an orthogonal decomposition for any suspicious encoding:
\begin{equation}    \tilde{s}_i^\Delta = W \alpha_i^\Delta + e_i^\Delta, \label{orthogonal_decomposition}
\end{equation}
where $e_i^\Delta \perp \mathrm{col}(W)$, $\alpha_i^\Delta$ represents the coordinates of $\tilde{s}_i^\Delta$ in the legitimate subspace $\mathrm{col}(W)$, while the residual $e_i^\Delta$ quantifies the deviation from this subspace, satisfying $e_i^\Delta = 0$ for benign encodings and $e_i^\Delta \neq 0$ for poisoned ones.

Based on this decomposition, the projected representation is given by:
\begin{equation}
    y_i^\Delta = \Phi \tilde{s}_i^\Delta = \Phi W \alpha_i^\Delta + \Phi e_i^\Delta.
\end{equation}
where $\Phi e_i^\Delta$ captures the projection of structural deviation and serves as the geometric signal of poisoning. 

Since $\Phi$ is linear, its effect on residual amplification is bounded as:
\begin{equation}
    c_1 \|e_i^\Delta\|^2 \le \|\Phi e_i^\Delta\|^2 \le c_2 \|e_i^\Delta\|^2,
\end{equation}
where $c_1, c_2 > 0$ are constants determined by the extremal singular values of $\Phi$. If $\Phi$ satisfies the RIP condition, this deviation is preserved in the projection, enabling stable separation of inconsistent samples.

In summary, PEEL defines structural consistency through lossless reconstruction and exposes poisoning behaviors as deviations from this constraint. This model-agnostic mechanism provides a robust and geometry-aware foundation for poisoning exposure under LDP.

\section{Theoretical Guarantees of PEEL}

This section establishes the theoretical foundations of PEEL. On the client side, all PEEL encoding components are post-processing transformations applied only to the $\varepsilon$-LDP reports $z_i$, and they never re-access the raw sensitive data $x_i$. Thus, the transmission object $y_i$ inherits the same $\varepsilon$-LDP guarantee as $z_i$ by the post-processing property of differential privacy~\cite{dwork2006calibrating}, and we therefore do not elaborate further. With respect to statistical utility, we study PEEL under standard LDP mechanisms and prove preservation of the baseline LDP estimator properties. We establish two results: (i) Unbiasedness—estimators based on PEEL-decoded data have the same expectation as those built directly from the original LDP reports; and (ii) Statistical Accuracy—for native statistical queries, PEEL does not worsen established error bounds.

\subsection{Unbiasedness}
\label{Unbiased_analysis}

Let the front-end $\varepsilon$-LDP mechanism $\psi_\varepsilon$ produce privatized reports $z_{1:n}$. Define the baseline estimator:
\begin{equation}
    \widehat{\theta}_{\mathrm{LDP}} \coloneqq Q\big(t(z_1),\ldots,t(z_n)\big).\label{baseline_estimator}
\end{equation}

We consider aggregation operators $Q$ that are linear or dimension-wise (e.g., dimension-wise counts, frequencies, or means). For such $Q$, the baseline estimator is unbiased, i.e., $\mathbb{E}[\widehat{\theta}_{\mathrm{LDP}}]=\theta$.

The sparsification map $\mathrm{S}$ is chosen to satisfy the conditional-alignment property \eqref{eq:HT-alignment}, followed by per-data, across-dimensions z-score normalization \eqref{z-score_data} and a low-rank projection \eqref{Structural_consistency_expression}. On the receiver side, a linear decoder matched to the encoding map implements both projection inversion and z-score restoration—i.e., it performs closed-loop reconstruction \eqref{eq:closed_loop} and then restores the canonical 1-sparse form—so that $y_i$ is mapped directly to $s_i$. This pipeline allows us to establish the following unbiasedness guarantee.

\begin{theorem}[Unbiasedness Preservation]
If $Q$ is linear or dimension-wise, the PEEL-integrated LDP estimator
\begin{equation}
    \widehat{\theta}_{\mathrm{PEEL}}(\hat{s}_{1:n}) \coloneqq Q(\hat{s}_1,\ldots,\hat{s}_n) \label{PEEL-integrated_LDP_estimator}
\end{equation}
satisfies:
\begin{equation}
\mathbb{E}\!\left[\widehat{\theta}_{\mathrm{PEEL}}(\hat{s}_{1:n})\right]
=
\mathbb{E}\!\left[\widehat{\theta}_{\mathrm{LDP}}(z_{1:n})\right]
=\theta. \label{Unbiasedness_Preservation}
\end{equation}
\end{theorem}

\begin{proof}
By the closed-loop reconstruction \eqref{eq:perfect_reconstruction} and the deterministic restore operator \eqref{restore}, we have
\begin{equation}    \mathbb{E}\big[Q(\hat{s}_{1:n})\big]=\mathbb{E}\big[Q(s_{1:n})\big].
\end{equation}

For each $i$, the alignment condition \eqref{eq:HT-alignment} together with the law of iterated expectations gives:
\begin{equation}
    \mathbb{E}[s_i]=\mathbb{E}\!\big[\mathbb{E}[s_i\mid z_i]\big]=\mathbb{E}[t(z_i)] .
\end{equation}

Since $Q$ is linear or dimension-wise, expectation commutes with aggregation, hence:
\begin{equation}
    \mathbb{E}\big[Q(s_{1:n})\big]
=\mathbb{E}\big[Q\big(t(z_{1:n})\big)\big]
=\mathbb{E}\big[\widehat{\theta}_{\mathrm{LDP}}\big]
=\theta.
\end{equation}

For LDP mechanisms whose client reports are inherently 1-sparse with symmetric signs and known selection probabilities, the alignment condition \eqref{eq:HT-alignment} is satisfied without further transformation. Consequently, unbiasedness is inherited directly. For mechanisms with multi-dimensional or dense outputs, the Horvitz--Thompson sparsification in \eqref{eq:ht-construction} preserves the per-data expectation, and thus yields the same unbiasedness as the baseline for the above classes of $Q$.
\end{proof}

\subsection{Statistical Accuracy}
\label{sec:statistical_accuracy}

This section proves that integrating PEEL does not degrade statistical accuracy for linear or dimension-wise aggregation queries. Let the sample-level contribution of the aggregator $Q$ be expressible as:
\begin{equation}
    Q(s_{1:n}) = \frac{1}{n}\sum_{i=1}^n q(s_i),
\end{equation}
where $q(\cdot)$ is linear or dimension-wise additive. 

For linear/dimension-wise $Q$, closed-loop reconstruction and deterministic restoration \eqref{eq:closed_loop} have:
\begin{equation}
    Q(\hat{s}_{1:n}) \,=\, Q(s_{1:n}).
\end{equation}

Hence, any potential accuracy difference can only arise from the sampling randomness in the sparsification map $\mathrm{S}$. For i.i.d. client data, the law of total variance yields:
\begin{equation}
\!\mathrm{Var}\big(Q(s_{1:n})\big)\!=\!\mathrm{Var}\big(Q(t(z_{1:n}))\big)
\!+\!\frac{1}{n}\,\mathbb{E}\big[\mathrm{Var}\big(q(s_i)\!\mid\! z_i\big)\big],\!
\label{eq:var-decomp-keep}
\end{equation}
where the first equality uses the alignment \eqref{eq:HT-alignment} and the linear/dimension-wise structure of $Q$ to pass conditional expectations through samples/dimensions. 

Since unbiasedness has been established in \eqref{Unbiasedness_Preservation}, MSE equals variance, and thus:
\begin{equation}
\mathrm{MSE}\big(\widehat{\theta}_{\mathrm{PEEL}}\big)
= \mathrm{MSE}\big(\widehat{\theta}_{\mathrm{LDP}}\big) + \Delta_n,
\label{eq:MSE-decomp-keep}
\end{equation}
where $ \Delta_n \triangleq \frac{1}{n}\,\mathbb{E}\big[\mathrm{Var}\big(q(s_i)\mid z_i\big)\big]\ge 0$.

\begin{theorem}[Accuracy Preservation]If client reports are naturally 1-sparse with symmetric signs and known selection probabilities, then
\begin{equation}
    \Delta_n = 0, \mathrm{MSE}\big(\widehat{\theta}_{\mathrm{PEEL}}\big)=\mathrm{MSE}\big(\widehat{\theta}_{\mathrm{LDP}}\big).
\end{equation}
If reports are multi-dimensional/dense and the sparsification uses the Horvitz--Thompson construction \eqref{eq:ht-construction}, and if the sample-level contribution for linear/dimension-wise aggregation is $q(s_i)=\sum_j w_j s_{i,j}$, then
\begin{equation}
    \Delta_n
= \frac{1}{n}\,\mathbb{E}\left[\sum_{j} w_j^2\,t_j(z_i)^2\left(\frac{1}{p_j(z_i)}-1\right)\right].
\end{equation}

Under the constraint \(\sum_j p_j(z_i)=1\), choosing
\begin{equation}
    p_j^\star(z_i)\ \propto\ |w_j\,t_j(z_i)|,
\end{equation}
minimizes the additive term above, thereby retaining the baseline \(O(1/n)\) error rate and the same \((\varepsilon,k)\)-dependence; the additive contribution is constant-order and can be optimized via the sampling allocation.
\end{theorem}

\begin{proof}
If client reports are naturally 1-sparse with symmetric signs and known selection probabilities, then $s_i$ carries no additional randomness given $z_i$, hence:
\begin{equation}
    \mathrm{Var}\big(q(s_i)\mid z_i\big)=0
\Rightarrow
\Delta_n=0,
\end{equation}
and therefore:
\begin{equation}    \mathrm{MSE}\big(\widehat{\theta}_{\mathrm{PEEL}}\big)=\mathrm{MSE}\big(\widehat{\theta}_{\mathrm{LDP}}\big).
\end{equation}

Thus, for these mechanisms, PEEL preserves the baseline statistical accuracy (no degradation).

If client reports are multi-dimensional/dense, adopt the Horvitz--Thompson sparsification, i.e., \eqref{eq:ht-construction}. For each coordinate $j$ (conditional on $z_i$) have:
\begin{equation}
\!\mathbb{E}[s_{i,j}\!\mid\! z_i]\!=\!t_j(z_i),
\mathrm{Var}(s_{i,j}\!\mid\! z_i)\!=\!t_j(z_i)^2\left(\frac{1}{p_j(z_i)}\!-\!1\right)\!.\!\!\!
\label{eq:HT-cond-var-keep}
\end{equation}

Let $q(s_i)=\sum_j w_j s_{i,j}$ denote the sample-level contribution for linear/dimension-wise aggregation (e.g., $w_j=1$ for per-dimension means/frequencies). Then, substituting \eqref{eq:ht-construction} into \eqref{eq:HT-cond-var-keep} yields:
\begin{equation}
\mathrm{Var}\big(q(s_i)\mid z_i\big)
=\sum_{j} w_j^2\,t_j(z_i)^2\left(\frac{1}{p_j(z_i)}-1\right),
\end{equation}
and substituting into \eqref{eq:MSE-decomp-keep} gives:
\begin{equation}
\Delta_n
=\frac{1}{n}\,\mathbb{E}\left[\sum_{j} w_j^2\,t_j(z_i)^2\left(\frac{1}{p_j(z_i)}-1\right)\right].
\label{eq:Delta-weighted-keep}
\end{equation}

Under $\sum_j p_j(z_i)=1$, minimizing $\sum_{j} \frac{w_j^2\, t_j(z_i)^2}{p_j(z_i)}$ yields:
\begin{equation}
p_j^\star(z_i)\ \propto\ |w_j\,t_j(z_i)|.
\label{eq:pstar-weighted-keep}
\end{equation}

With \eqref{eq:pstar-weighted-keep}, yields:
\begin{equation}
    \sum_j \frac{w_j^2 t_j(z_i)^2}{p_j^\star(z_i)}
=\Big(\sum_j |w_j t_j(z_i)|\Big)^2.
\end{equation}

So the optimal additive term admits the bound:
\begin{equation}
\Delta_n^\star
=\frac{1}{n}\,\mathbb{E}\Big[\ \big\|W\,t(z_i)\big\|_1^2 - \big\|W\,t(z_i)\big\|_2^2\ \Big],
\label{eq:Delta-opt-keep}
\end{equation}
where $W=\mathrm{diag}(w_1,\ldots,w_k)$. In particular, for $w_j\equiv 1$,
\begin{equation}
\Delta_n^\star
=\frac{1}{n}\,\mathbb{E}\big[\|t(z_i)\|_1^2-\|t(z_i)\|_2^2\big]
\le\frac{k-1}{n}\,\mathbb{E}\big[\|t(z_i)\|_2^2\big],
\end{equation}
where the final inequality uses $\|u\|_1^2\le k\|u\|_2^2$. Consequently, $\Delta_n=O(1/n)$, preserving the $1/n$ error order and the same $(\varepsilon,k)$ dependence as the baseline; the additive constant can be optimized via \eqref{eq:pstar-weighted-keep}. When $t(z_i)$ is itself (approximately) 1-sparse, $\|t\|_1^2\approx\|t\|_2^2$ and the additive term is negligible.
\end{proof}

These results show that, for LDP mechanisms that naturally produce 1-sparse reports, statistical accuracy is preserved. For non–1-sparse mechanisms, PEEL contributes at most an optimizable constant-order additive term, while retaining the baseline $O(1/n)$ error rate and the same $(\varepsilon,k)$-dependence. Overall, PEEL enables structural reconstruction and consistency checking without degrading the statistical accuracy of LDP-based analyses.

\section{Robustness Analysis of PEEL}

This section establishes a theoretical framework for analyzing the robustness of PEEL, focusing on its ability to expose poisoning behaviors through structural consistency. PEEL achieves poisoning exposure by amplifying structural inconsistencies in a constrained geometric space. Under 1-sparse encoding and symmetric normalization, benign samples yield standardized representations confined to two discrete magnitudes. Any reconstruction that falls outside this support indicates structural deviation and reveals potential poisoning.

\subsection{Against Output Poisoning}

We consider the normalized structural reference matrix $\tilde{\mathcal{D}} \in \mathbb{R}^{k \times 2k}$, with columns drawn from the discrete set $\{\pm v_1, \pm v_2\}$ for some $v_1, v_2 \in \mathbb{R}^+$. Accordingly, each normalized sparse structural vector $\tilde{s}_i \in \mathbb{R}^k$ has at most one nonzero entry, satisfying $\tilde{s}_{i,j} \in \{\pm v_1, \pm v_2\}, \|\tilde{s}_i\|_0 \le 1$. Define index sets $\Omega_1 = \{ j \mid \tilde{s}_{i,j} = v_1 \}$,$\Omega_2 = \{ j \mid \tilde{s}_{i,j} = -v_2 \}$. Then $\tilde{s}_i$ can be represented as a sparse linear combination of canonical basis vectors:
\begin{equation}
    \tilde{s}_i = v_1 \sum_{j \in \Omega_1} w_j - v_2 \sum_{j \in \Omega_2} w_j,
\end{equation}
where $w_j \in \mathbb{R}^k$ denotes the $j$-th standard basis vector.

After projection, the structural expression becomes:
\begin{equation}
y_i = \Phi \tilde{s}_i = v_1 \sum_{j \in \Omega_1} \Phi w_j - v_2 \sum_{j \in \Omega_2} \Phi w_j.
\end{equation}

Under an output poisoning attack, adversaries inject a perturbation $\Delta$ into the projected vector:
\begin{equation}
y^{\Delta}_i = y_i + \Delta = v_1\sum_{j \in \Omega_1} \Phi w_j -v_2\sum_{j \in \Omega_2} \Phi w_j + \Delta,
\end{equation}

The poisoned reconstruction is then computed as:
\begin{equation}
\hat{s}^{\Delta}_i = \Gamma y^{\Delta}_i = W \Theta^\dagger y^{\Delta}_i = W \Theta^\top y^{\Delta}_i,
\end{equation}

Since $\Theta^\top$ is symmetric and positive semidefinite, it admits an eigen decomposition $\Theta^\top = U \Lambda U^\top$, with orthonormal eigenvectors $\{u_\ell\}_{\ell=1}^{k-1}$, yielding
\begin{equation}
\begin{aligned}
\hat{s}^{\Delta}_i &= W \Theta^\top \Bigg( v_1\sum_{j \in \Omega_1} \Phi w_j 
        - v_2\sum_{j \in \Omega_2} \Phi w_j + \Delta \Bigg) \label{eq:poisoned_s}\\
& = W \sum_{\ell=1}^{k-1} \Bigg( v_1 \sum_{j \in \Omega_1} u_\ell^\top \Phi w_j 
        - v_2\sum_{j \in \Omega_2} u_\ell^\top \Phi w_j + u_\ell^\top \Delta \Bigg) u_\ell 
\end{aligned}
\end{equation}
Here, the first two terms are fixed and deterministic, while the term $u_\ell^\top \Delta$ incurs unpredictable shifts caused by poisoning noise. Given that the dimension of $\mathrm{null}(\Theta^\top)$ is typically small, $u_\ell^\top \Delta \ne 0$ holds with high probability. This transformation converts even small poisoning noise into continuous deviations that force $\hat{s}_i^\Delta$ to leave its discrete domain, thereby amplifying the observable reconstruction error. Thus, the reconstructed vector deviates from the expected discrete domain, resulting in observable structural inconsistency.

\subsection{Against Rule Poisoning}

This section investigates the robustness of PEEL under rule poisoning attacks, where adversaries do not tamper with individual samples directly but instead manipulate system-level parameters such as the privacy budget $\varepsilon$ or the projection matrix $\Phi$. Unlike explicit data corruption, such attacks operate at the mechanism layer and exhibit higher stealthiness and transferability, making them difficult to detect using traditional sample-level defenses.

\subsubsection{Under Privacy-Budget Poisoning}

This subsection theoretically examines PEEL's robustness against privacy budget poisoning, wherein adversaries manipulate the privacy parameter $\varepsilon^\Delta \ne \varepsilon$ to induce systematic deviations in perturbation strength. A poisoned budget $\varepsilon^\Delta < \varepsilon$ intensifies noise, while $\varepsilon^\Delta > \varepsilon$ reduces it, both scenarios may destabilize the structural encoding and introduce poisons in the projected space.

Under a benign configuration, the perturbed output $z_i$ is mapped to a standardized structural encoding $\tilde{s}_i \in \mathrm{col}(W)$, remaining within the structural subspace. Under poisoning, the reconstructed structure is computed as:
\begin{equation}
    \hat{s}_i^\Delta = \Gamma \Phi \tilde{s}_i^\Delta.
\end{equation}

Based on the orthogonal decomposition in \eqref{orthogonal_decomposition}, any vector $\tilde{s}_i^\Delta$ can be expressed as a sum of its projection onto the structural subspace $\mathrm{col}(W)$ and an orthogonal residual. Since the reconstruction $\hat{s}_i^\Delta$ corresponds exactly to this projection, we have $\hat{s}_i^\Delta = WW^\top \tilde{s}_i^\Delta$. We thus define the structural consistency error as follows:
\begin{equation}
    \delta_i^\Delta := \|\hat{s}_i^\Delta - \tilde{s}_i^\Delta\|
\end{equation}
which quantifies the deviation of reconstruction from its expected discrete form. This error can be written as:
\begin{equation}
\delta_i^\Delta = \|W W^\top \tilde{s}_i^\Delta - \tilde{s}_i^\Delta\| 
= \|(I - \Pi_W) \tilde{s}_i^\Delta\|, 
\label{eq:residual_projection}
\end{equation}
where $\Pi_W = WW^\top$ denotes the orthogonal projection operator onto $\mathrm{col}(W)$, eliminating any component outside the structural subspace.

Consider the perturbed model in \eqref{Structured_coding_random_perturbation}, where the poisoning-induced perturbation $\mathcal{R}_i^{(\varepsilon^\Delta)}$ is assumed to be a sub-Gaussian random vector. If the structural encoding process is linear, then the resulting encoded vector $\tilde{s}_i^\Delta$ also follows a sub-Gaussian distribution. Applying the projection $(I - \Pi_W)\tilde{s}_i^\Delta$ yields a linear transformation of a sub-Gaussian vector, which remains sub-Gaussian. Consequently, the residual norm $\delta_i^\Delta$ becomes a sub-exponential random variable.

Let $u_i^\perp := (I - \Pi_W)\tilde{s}_i^\Delta$ denote the structural deviation vector, with covariance matrix $\Sigma^\Delta := \mathrm{Cov}(u_i^\perp)$ and spectral norm $\|\Sigma^\Delta\|$. By leveraging sub-exponential concentration inequalities~\cite{vershynin2018high, wainwright2019high}, the tail probability of $\delta_i^\Delta$ admits the following bound:
\begin{equation}
    \mathbb{P}(\delta_i^\Delta > \tau) \le 2 \exp\left( -c \cdot \min\left\{ \frac{\tau^2}{\|\Sigma^\Delta\|}, \frac{\tau}{\sqrt{\|\Sigma^\Delta\|}} \right\} \right),
    \label{eq:tail_bound}
\end{equation}
where $c$ is a universal constant.

Define a confidence bound $\tau_\varepsilon$ under benign noise:
\begin{equation}
    \mathbb{P}(\delta_i > \tau_\varepsilon \mid \varepsilon) \le \alpha, \quad
    \tau_\varepsilon := \sqrt{ \tfrac{\sigma_\varepsilon^2}{c} \log\!\left( \tfrac{2}{\alpha} \right) }.
\end{equation}

This bound establishes a reference region where deviations are statistically negligible with high confidence.

\textbf{Case 1: $\varepsilon^\Delta < \varepsilon$ (Excessive Noise).}  Here, the effective variance of $\delta_i^\Delta$ increases, so that
\begin{equation}
    \mathbb{P}(\delta_i^\Delta > \tau_\varepsilon \mid \varepsilon^\Delta) \gg \alpha,
\end{equation}
indicating frequent violations of the baseline confidence region.

\textbf{Case 2: $\varepsilon^\Delta > \varepsilon$ (Weakened Noise).}  Although nominal variance shrinks, poisoning incurs orthogonal deviations not modeled by $\varepsilon$. Consequently,
\begin{equation}
    \mathbb{P}(\delta_i^\Delta > \tau_\varepsilon \mid \varepsilon^\Delta) > \alpha,
\end{equation}
indicating structural shifts beyond the benign baseline.

These results suggest that even when raw data remain untouched, privacy-budget poisoning yields measurable structural deviation in the PEEL projection-reconstruction pipeline. PEEL can expose such deviations by evaluating the probability $\mathbb{P}(\delta_i^\Delta > \tau_\varepsilon)$ against the confidence threshold $\tau_\varepsilon$.

\subsubsection{Under Projection Matrix Poisoning}

Consider a scenario where the projection matrix is poisoned and is independent of the original projection matrix $\Phi$, i.e.,
\begin{equation}
    \Phi_{\mathrm{poisoned}} = \Phi + \Delta.
\end{equation}

Under poisoning, the received measurement is modified as:
\begin{equation}
    y_i^\Delta = \Phi_{\mathrm{poisoned}} \tilde{s}_i = (\Phi + \Delta) \tilde{s}_i.
\end{equation}

The receiver side reconstructs the structure using the inverse mapping $\Gamma$ (as defined in \eqref{inverse_transformation}) based on the unpoisoned matrix:
\begin{equation}
    \hat{s}_i^\Delta = \Gamma y_i^\Delta = \Gamma (\Phi + \Delta) \tilde{s}_i.
\end{equation}

Expanding the $i$-th component of the reconstructed vector:
\begin{equation}
    \hat{s}_i^\Delta = \sum_{p=1}^{k-1} \Gamma_{ip} y_p^\Delta 
    = \sum_{p=1}^{k-1} \sum_{j=1}^k \Gamma_{ip} (\phi_{pj}^{\mathrm{true}} + \delta_{pj}) \tilde{s}_j
    = E_i + P_i,
\end{equation}
where $E_i = \sum_{p,j} \Gamma_{ip} \phi_{pj}^{\mathrm{true}} \tilde{s}_j$ represents the nominal (benign) structural component, and $P_i = \sum_{p,j} \Gamma_{ip} \delta_{pj} \tilde{s}_j$ denotes the poisoning-induced deviation.

Since $\tilde{s}_i$ is a sparse and normalized structural vector whose entries are restricted to the discrete set $\{\pm v_1, \pm v_2\}$, the aggregator can reliably recover its structural state under benign conditions. If adversaries attempt to manipulate the output so that the reconstructed vector $\hat{s}_i^\Delta$ no longer corresponds to its true value $\tilde{s}_i$ but is instead forced to align with another discrete point $v_s^\Delta \in \{\pm v_1, \pm v_2\}$ (different from its original assignment), then the injected perturbation must satisfy:
\begin{equation}
    P_i = v_s^\Delta - E_i, \label{injected_perturbation_must_satisfy}
\end{equation}
where $P_i$ represents a crafted offset. 

Unlike random noise, it is not restricted to follow any prescribed distribution, and its construction depends entirely on the adversary’s strategy. However, because $P_i$ must exactly cancel and replace the discrete value of $E_i$, the feasibility of achieving \eqref{injected_perturbation_must_satisfy} corresponds to hitting a single point in a continuous space. Consequently, the probability of exact alignment is zero:
\begin{equation}
    \mathbb{P}(P_i = v_s^\Delta - E_i) = 0.
\end{equation}

Furthermore, the probability of simultaneously achieving precise control over all $k$ components, such that a new discrete pattern is formed (e.g., $1:(k{-}1)$ ratio), is given by:
\begin{equation}
    \mathbb{P}\left( \bigcap_{i=1}^k \{P_i = v_s^\Delta - E_i\} \right) = \prod_{i=1}^k \mathbb{P}(P_i = v_s^\Delta - E_i) = 0,
\end{equation}
which constitutes a Lebesgue-null set in the probability space and is thus almost surely unachievable.

In summary, when the aggregator reconstructs the structural representation along the legitimate decoding path, it becomes statistically infeasible for the adversary to deterministically steer the output $\hat{s}_i^\Delta$ through the injection of a poisoning matrix $\Delta$, as the reconstruction is constrained to a fixed discrete structural domain. Specifically, it is statistically infeasible to align the corrupted output with a new, consistent discrete pattern. As a result, the reconstructed vector $\hat{s}_i^\Delta$ deviates from the original discrete set $\{\pm v_1, \pm v_2\}$, leading to either multi-valued outputs or boundary drift. These deviations inherently violate the discrete structural constraints and thus serve as indicators of poisoning.

\section{Performance and Experimental Analysis}

This section evaluates PEEL from both performance and experimental perspectives. The performance analysis quantifies the client-side overhead in the encode-transmit pipeline under LDP constraints. The experimental analysis assesses end-to-end effectiveness on real IoT datasets—typical LDP deployment environments—using standard LDP mechanisms and attack models to evaluate poisoning exposure utility and robustness. Together, these analyses demonstrate that PEEL is both effective and lightweight in practice. 

\subsection{Performance Analysis}

LDP is predominantly deployed in decentralized and resource-constrained IoT settings, such as smart grids and vehicular networks \cite{YANG2024103827, shamshad_provably_2023, BATOOL2024119717}, where trusted aggregation is unavailable and per-client efficiency is essential. In these scenarios, privacy-preserving mechanisms must enforce rigorous protection while minimizing computational and communication costs. As a structural poisoning exposure framework tailored for LDP, PEEL must adhere to these constraints to ensure practical deployability. Accordingly, our analysis considers both computation and communication overhead on the client-side, reflecting the key efficiency requirements in LDP-based systems.

To support a comprehensive and fair evaluation, we consider two categories of representative privacy-preserving mechanisms. The first category comprises general-purpose privacy-preserving collaboration mechanisms, including Federated Learning  (FL) frameworks~\cite{badr_privacy-preserving_2023} and cryptography-based secure aggregation protocols~\cite{shamshad_provably_2023,parameswarath_privacy-preserving_2023}, which serve as strong baselines in distributed settings. The second category consists of poisoning-resilient pre-perturbation defenses under LDP constraints, including Secure OLH~\cite{10.1007/978-3-030-81242-3_3}, VGRR~\cite{10220122}, emPrivKV~\cite{app14146368}, and OT-HCMS~\cite{10.1007/978-3-031-68208-7_18}, which provide localized robustness through statistical filters, cryptographic commitments, or oblivious transfer (OT) protocols. These mechanisms capture both the state-of-the-art in secure aggregation and the current landscape of poisoning mitigation techniques in LDP settings.

Our experimental implementation of PEEL builds upon Harmony~\cite{nguyen2016collecting}, an LDP mechanism used by Samsung for smartphone telemetry collection. This choice is motivated by Harmony's inherent generation of 1-sparse outputs through random dimension selection, which naturally aligns with PEEL's structural condition (C1). Harmony serves as a representative instance of dimension-selection-based LDP mechanisms (including Duchi~\cite{duchi2018minimax}, PM~\cite{wang2019collecting}), demonstrating PEEL's compatibility with this class of methods. While our evaluation focuses on Harmony for clarity, PEEL's framework readily extends to other LDP mechanisms satisfying the structural conditions.

\subsubsection{Communication Overhead}

Communication overhead is a critical factor in evaluating the deployability of privacy-preserving mechanisms, especially in distributed environments where wireless transmission is dominant and data transfer costs can far exceed local computation. Transmitting a single bit consumes over 1000 times the energy required for a 32-bit arithmetic operation~\cite{barr2006energyaware}. This section compares the communication cost of Harmony-integrated PEEL (Harmony-PEEL) against representative mechanisms, as outlined in Table~\ref{tab:communication_overhead_detailed}.

\begin{table*}[hbtp]
    \centering
    \caption{Comparison of Client-Side Communication Overhead (Per Round)}
    \label{tab:communication_overhead_detailed}
    \begin{tabular}{ccc}
        \hline
        \textbf{Scheme} & \textbf{Transmitted Content (per client)} & \textbf{Total Overhead (bits)} \\ \hline
        Badr et al.~\cite{badr_privacy-preserving_2023} 
        & $n$ float32 parameters ($\approx$ 1 KB/parameter) after CAT filtering 
        & $\geq 466,944$\\

        Shamshad et al.~\cite{shamshad_provably_2023} 
        & ECC public key (608) + ECC ciphertext (1280) + AES payload (128) 
        & 2016 \\ 

        Parameswarath et al.~\cite{parameswarath_privacy-preserving_2023} 
        & RSA-auth token (1760) + signature (2272)& $\geq$ 4032 \\ 

        emPrivKV~\cite{app14146368} 
        & 5 rounds of 1-out-of-$d$ OT, each sending a 2048-bit ciphertext 
        & $5 \cdot \lceil \log_2 d \rceil \cdot 2048$ \\ 

        VGRR~\cite{10220122} 
        & $\ell$ Pedersen commitments (2048-bit each) + openings for $1 + d\ell_2$ slots 
        & $2 \cdot \ell \cdot 2048$ (worst-case) \\ 

        Secure OLH~\cite{10.1007/978-3-030-81242-3_3} 
        & $n$ commitments + $g$ encoded slots, each 2048-bit 
        & $(2n + g) \cdot 2048$ \\ 

        OT-HCMS~\cite{10.1007/978-3-031-68208-7_18} 
        & 4 OT ciphertexts (2048-bit) + 1 hashed index (32-bit) + 1-bit response 
        & 8209 \\ 

        Harmony-PEEL 
        & $(k{-}1)$ projected dimensions, each encoded with $\lceil \log_2(k{-}1) \rceil$ bits 
        & 2016 (For $k = 252$)\\ \hline
    \end{tabular}
\begin{minipage}{\textwidth}
    \footnotesize\raggedright
    \textbf{Note:} All values represent per-round communication cost. Here, $n$ is the model dimensionality, $d$ is the domain size, $g = \lceil d/2 \rceil$ is the hash output dimension in Secure OLH, and $k$ is the number of encoding bins in Harmony-PEEL. Communication parameters are standardized: ECC keys are 224-bit, RSA and OT messages are 2048-bit, and all hash outputs are 256-bit.
\end{minipage}
\end{table*}

For representative privacy-preserving mechanisms, we examine three baselines. Badr et al.~\cite{badr_privacy-preserving_2023} proposed an FL scheme using categorical adaptive thresholding (CAT) to filter low-impact updates. Assuming $n$ model parameters and 1 KB per parameter, the total communication per round exceeds 57 KB. Shamshad et al.~\cite{shamshad_provably_2023} designed a three-party protocol integrating ECC and AES, transmitting 2016 bits per session per client. Parameswarath et al.~\cite{parameswarath_privacy-preserving_2023} further introduced a zero-knowledge proof (ZKP)-enhanced authentication protocol involving RSA tokens and signatures, yielding at least 4032 bits per session (excluding ZKP expansion).

For poisoning-resilient pre-perturbation defenses, we analyze four representative mechanisms. emPrivKV~\cite{app14146368} protects access patterns through five rounds of 1-out-of-$d$ OT, where each round transmits a 2048-bit ciphertext corresponding to a securely retrieved key. This design avoids explicit perturbation while maintaining privacy and estimation utility. VGRR~\cite{10220122} employs $\ell$ Pedersen commitments to bind the structure of local outputs, followed by opening up to $1 + d\ell_2$ values for server-side integrity verification. This scheme enforces structural accountability with minimal client-side cost. Secure OLH~\cite{10.1007/978-3-030-81242-3_3} enhances the OLH mechanism by encrypting both the encoded slots and perturbation masks, yielding $(2n + g)$ ciphertexts of 2048 bits each; it further supports verifiability via zero-knowledge proofs. OT-HCMS~\cite{10.1007/978-3-031-68208-7_18} combines Hadamard sketching with OT-based noise injection, where each client transmits $2^\tau$ OT ciphertexts (2048 bits each), along with a hashed index (32 bits) and a binary value, totaling 8209 bits under $\varepsilon=1$.

In our Harmony-based instantiation of PEEL, Harmony encodes each input as a 1-sparse vector with a deterministically chosen non-zero dimension.  PEEL exploits this structure to apply projection-based encoding, yielding a projected vector $\hat{y} \in \mathbb{R}^{k{-}1}$.  Each dimension is discretized into $\lceil \log_2(k{-}1) \rceil$ bits, leading to a total communication cost of $(k{-}1) \cdot \lceil \log_2(k{-}1) \rceil$ bits.  For $k = 252$, the results in 2016 bits, which is consistent with the overhead reported in~\cite{shamshad_provably_2023}; smaller $k$ further reduces the communication overhead.

Harmony-PEEL achieves superior communication efficiency compared to the FL frameworks and cryptographic-heavy authentication protocols~\cite{badr_privacy-preserving_2023, shamshad_provably_2023, parameswarath_privacy-preserving_2023}, which incur kilobyte-level or multi-round costs. It also outperforms state-of-the-art LDP poisoning defenses~\cite{app14146368, 10220122, 10.1007/978-3-030-81242-3_3, 10.1007/978-3-031-68208-7_18} that involve ciphertext transmission, commitment proofs, or ZKP validation. By ensuring bit-level compactness while preserving structural integrity and $\varepsilon$-LDP guarantees, Harmony-PEEL offers high deployability in bandwidth-sensitive, wireless, and resource-constrained environments.

\subsubsection{Computation Overhead}

As LDP mechanisms operate entirely on the client side, computation efficiency is a critical factor for practical deployment in resource-constrained environments.  We evaluate the per-client computation overhead of PEEL in terms of runtime latency per round, which directly reflects the feasibility of integration into large-scale data collection systems. The result is shown in Table~\ref{tab:computation_overhead_quantified}.

\begin{table*}[hbtp]
    \centering
    \caption{Comparison of Client-Side Computation Overhead (Per Round)}
    \label{tab:computation_overhead_quantified}
    \begin{tabular}{ccc}
        \hline
        \textbf{Scheme} & \textbf{Client Operations} & \textbf{Estimated Time/Client (ms)}\\ \hline
        Badr et al.~\cite{badr_privacy-preserving_2023} & $O(n)$ gradient filtering & $\approx$ 1\\ 
        Shamshad et al.~\cite{shamshad_provably_2023} & 2 ECC + 2 Hash + 1 AES & $\approx$2\\ 
        Parameswarath et al.~\cite{parameswarath_privacy-preserving_2023} & 2 RSA + 3 Hash + 1 ZKP Gen + 1 SigVerify & $\approx$121\\ 
        emPrivKV~\cite{app14146368} & $5 \cdot \lceil \log_2 d \rceil$ OT encryptions & $\approx$ $100 \cdot d$\\ 
        VGRR~\cite{10220122} & $\ell$ Pedersen Commit + $1{+}d\ell_2$ open ops & $\approx 20 \cdot \ell$\\ 
        Secure OLH~\cite{10.1007/978-3-030-81242-3_3} & $n$ commitments + $g$ proofs (Pedersen + ZKP + Hash) & $\approx 10 \cdot n \cdot g$\\ 
        OT-HCMS~\cite{10.1007/978-3-031-68208-7_18} & $2^\tau$ OT (each with 1 enc + 2 dec) + 1 Hadamard & $\approx 30 \cdot 2^\tau$\\ 
        Harmony-PEEL & 1 Hash + 1 Projection & $\approx 0.01$\\ \hline
    \end{tabular}
    \begin{minipage}{\textwidth}
    \footnotesize\raggedright
    \textbf{Note:} AES and hash operations are about $1\,\mu$s per call, based on the \texttt{openssl speed} benchmark with AES-NI support~\cite{openssl-bench,aesni-benchmark}. RSA (2048-bit) and ECC (256-bit) are roughly $50$ ms and $1$ ms, respectively, according to OpenSSL on commodity Intel CPUs~\cite{openssl-bench}.  OT typically costs $20$–$50$ ms in implementations such as libOTe~\cite{libote-eprint}. Pedersen commitments and ZKP generation take about $1$ ms and $20$ ms, respectively, consistent with elliptic-curve and SNARK frameworks~\cite{bulletproofs,groth16}. Sparse projection adds $\approx 9\,\mu$s, as observed in BLAS/OpenBLAS benchmarks~\cite{blas-bench}. Overheads from memory allocation are excluded since LDP perturbations are applied immediately before release.
\end{minipage}
\end{table*}
To contextualize PEEL’s efficiency, we compare client-side computation overhead among three representative LDP baselines. In~\cite{badr_privacy-preserving_2023}, each client filters an $n$-dimensional gradient vector, requiring $O(n)$ comparisons and thresholding, but no cryptographic operations. Shamshad et al.~\cite{shamshad_provably_2023} combine lightweight primitives—two elliptic curve operations, two hash computations, and one AES encryption—yielding microsecond-level runtime. In contrast,~\cite{parameswarath_privacy-preserving_2023} involves high-cost primitives: RSA encryption, ZKP generation, and signature verification. These operations are unsuitable for frequent execution on constrained clients.

For poisoning-resilient pre-perturbation defenses, we evaluate four mechanisms. emPrivKV~\cite{app14146368} requires each client to perform $5 \cdot \lceil \log_2 d \rceil$ rounds of OT encryption, leading to a linear-time cost in $d$. VGRR~\cite{10220122} generates $\ell$ Pedersen commitments and opens up to $1 + d\ell_2$ slots, incurring moderate computation tied to cryptographic group operations. Secure OLH~\cite{10.1007/978-3-030-81242-3_3} augments OLH with commitment proofs and ZKPs across $n$ input dimensions and $g$ output components. Its cost scales with $O(n \cdot g)$ and becomes non-negligible for large $d$. OT-HCMS~\cite{10.1007/978-3-031-68208-7_18} applies Hadamard sketching and $2^\tau$ rounds of OT; with $\varepsilon = 1 \Rightarrow \tau = 2$, each round includes encryption and two decryptions. Although resilient to poisoning attacks, these protocols are computation-heavy.

Harmony-PEEL executes only one hash and one projection operation per round. The projection maps a $k$-dimensional one-hot vector into a $(k{-}1)$-dimensional space for exposure analysis, with total runtime below 10 $\mu$s. No asymmetric encryption, ZKP, or iterative interaction is involved. This low overhead allows Harmony-PEEL to remain scalable and responsive in bandwidth- and energy-constrained LDP settings.

Harmony-PEEL achieves the lowest client-side computation cost among all surveyed mechanisms. Its microsecond-level runtime outperforms cryptographic and hybrid protocols by 1–2 orders of magnitude, while preserving both privacy and poisoning exposure fidelity. This efficiency makes it well-suited for practical deployment in decentralized and resource-constrained data collection settings. 
\vspace{-1ex}

\subsection{Experimental Analysis}

\noindent\textbf{Environment.} Experiments were run in Python~3.9 on Windows~11 with an Intel Core i7-13700 (2.10\,GHz) and 16\,GB RAM. Source code is available at the project repository~\cite{shuai2025poisoncatcher}.

\noindent\textbf{Datasets.} We evaluate on two IoT datasets—the World Weather Repository (WWR)~\cite{global_weather_repository_kaggle} and the Smart Building Indoor Environmental dataset (SBD)~\cite{erol2023smartbuilding}. Preprocessing removes records with uneven spatial coverage or irregular sampling intervals; remaining numeric features are min–max scaled to [-1,1]. Variables unsuitable for mean/frequency analyses (e.g., overly dispersed or highly skewed) are filtered out.

\noindent\textbf{Parameters.} The LDP privacy budget is set to $\varepsilon=1$.  For the rule–poisoning attack, per-node budgets are sampled within the bounds specified by~\eqref{rule_poisoning_attacks}.  For the output–poisoning attack, outputs are randomized via a post-processing kernel constrained by~\eqref{output_poisoning_attacks}.

Most research on LDP poisoning has focused on defenses, with limited attention to identifying poisoned records. To date, three approaches estimate dataset-level poisoning ratios: DETECT~\cite{272214}, LDPGuard~\cite{10415225}, and a combined human and AI expert assessment (baseline). PoisonCatcher~\cite{shuai2025poisoncatcher} advances this line by statistically identifying poisoned records at the record level. Building on structural-consistency verification, PEEL enables precise record-level identification. To place all five methods on a common footing, this subsection evaluates their accuracy on the poisoning-ratio estimation task. For cross-query comparability, KRR is used as the LDP mechanism for frequency (categorical) queries, and the Laplace mechanism for mean (numeric) queries. Results are reported in Table~\ref{tab:attack_success}.

\begin{table*}[ht]
\centering
\caption{Comparison of Attack Ratio Estimates on WWR and SBD}
\label{tab:attack_success}
\renewcommand{\arraystretch}{1.0}
\setlength{\tabcolsep}{4.5pt} 

\begin{tabular}{@{}ccc@{\hskip 3pt}*{6}{cc}@{}}
\toprule
\textbf{Protocols} & \textbf{Attack Mode} & \makecell{\textbf{True} \\ \textbf{Attack Ratio}} 

& \multicolumn{2}{c}{\makecell{\textbf{DETECT} \\ \cite{272214} \textbf{Estimate}}} 
& \multicolumn{2}{c}{\makecell{\textbf{Expert} \\ \cite{shuai2025poisoncatcher} \textbf{Estimate}}} 
& \multicolumn{2}{c}{\makecell{\textbf{LDPGuard} \\ \cite{10415225} \textbf{Estimate}}} 
& \multicolumn{2}{c}{\makecell{\textbf{PoisonCatcher} \\ \cite{shuai2025poisoncatcher} \textbf{Identification}}}
& \multicolumn{2}{c}{\makecell{\textbf{PEEL} \\ \textbf{Identification}}} \\
\cmidrule(lr){4-5}
\cmidrule(lr){6-7}
\cmidrule(lr){8-9}
\cmidrule(lr){10-11}
\cmidrule(lr){12-13}
& & & \textbf{WWR} & \textbf{SBD}
      & \textbf{WWR} & \textbf{SBD}
      & \textbf{WWR} & \textbf{SBD}
      & \textbf{WWR} & \textbf{SBD}
      & \textbf{WWR} & \textbf{SBD} \\
\midrule

\multirow{2}{*}{Laplace} 
& Rule Poisoning Attack & 5\% & — & — & 8.66\% & 8.54\% & — & — & {5\%} & {4.94\%} & \textbf{5\%} & \textbf{5\%}\\
& Output Poisoning Attack & 5\% & — & — & 9.35\% & 8.83\% & — & — & {5\%} & {4.89\%} & \textbf{5\%} & \textbf{5\%}\\
\midrule

\multirow{2}{*}{KRR} 
& Rule Poisoning Attack & 5\% & — & — & 6.72\% & 6.31\% & — & — & {4.99\%} & {4.93\%} & \textbf{5\%} & \textbf{5\%}\\
& Output Poisoning Attack & 5\% & — & — & 8.62\% & 8.01\% & 38.2\% & 35.67\% & {5\%} & {4.96\%} & \textbf{5\%} & \textbf{5\%}\\
\bottomrule
\end{tabular}
\end{table*}

Across both datasets (WWR, SBD), both LDP mechanisms (Laplace for mean, KRR for frequency), and both attack modes (rule poisoning, output poisoning), PEEL’s attack-ratio estimate matches the ground truth (5\%) in every case. PoisonCatcher is the next best, staying within ±0.11 pp of the truth. In contrast, DETECT and the human+AI expert baseline systematically overestimate (e.g., 6.31–9.35\%), and LDPGuard is unstable—especially under KRR with output poisoning, where its estimates deviate drastically (35.67–38.2\%). These results show that PEEL’s structural consistency verification yields accurate and dataset/mechanism-agnostic attack-ratio estimates while retaining record-level localization capability.

\section{Conclusions}

PEEL leverages the intrinsic structural consistency of LDP encodings for poisoning exposure, operating as a post-processing module that requires no modification to existing LDP mechanisms. Theoretically, it preserves the unbiasedness and statistical accuracy of the underlying mechanism while exposing both output- and rule-level poisoning. Empirically, PEEL reduces client-side overhead compared to multiple privacy-preserving baselines and outperforms state-of-the-art defenses in poisoning-detection accuracy, demonstrating its practicality for large-scale IoT deployment.


%



\ifCLASSOPTIONcaptionsoff
  \newpage
\fi



\bibliographystyle{IEEEtran}
\bibliography{PEEL}
%

%




\end{document}